\documentclass[11pt]{article}

\usepackage{fullpage,appendix}
\usepackage[bookmarks,colorlinks,breaklinks]{hyperref}  % PDF hyperlinks, with coloured links
\hypersetup{linkcolor=blue,citecolor=blue,filecolor=blue,urlcolor=blue} % all blue links
\usepackage{amsmath,amsfonts,amsthm,amssymb,xspace,bm}
\newcommand{\newref}[2][]{\hyperref[#2]{#1~\ref*{#2}}}
\renewcommand{\eqref}[1]{\hyperref[#1]{(\ref*{#1})}}
\numberwithin{equation}{section}

\newcommand{\sref}[1]{\newref[Section]{sec:#1}}

\newcommand{\tref}[1]{\newref[Theorem]{thm:#1}}
\newcommand{\lref}[1]{\newref[Lemma]{#1}}

\newcommand{\cref}[1]{\newref[Corollary]{cor:#1}}

\theoremstyle{plain}
\newtheorem{theorem}{Theorem}[section]
\newtheorem{lemma}[theorem]{Lemma}
\newtheorem{claim}[theorem]{Claim}

\newtheorem{definition}[theorem]{Definition}

\theoremstyle{definition}

\newtheorem{remark}[theorem]{Remark}
\allowdisplaybreaks

\DeclareMathOperator*{\pr}{\mathsf{Pr}} 

\DeclareMathOperator*{\ex}{\mathbb{E}}
\newcommand{\rgta}{\rightarrow}
\newcommand{\lfta}{\leftarrow}

\newcommand{\reals}{\mathbb{R}}

\newcommand{\dpm}{\{1,-1\}}
\newcommand{\dpmn}{\dpm^n}
\newcommand{\sign}{\mathsf{sign}}
\newcommand{\sgn}{\mathsf{sgn}}

\newcommand{\infl}{\mathbb{I}}
\newcommand{\AS}{\mathbb{AS}}

\newcommand{\GNS}{\mathbb{GNS}}
\newcommand{\NS}{\mathbb{NS}}

\newcommand{\eps}{\varepsilon}
\renewcommand{\epsilon}{\varepsilon}

\newcommand{\note}[1]{\marginpar{\tiny *note in TeX*}}
\newcommand{\ignore}[1]{}

\newcommand{\calN}{{\cal N}}

\renewcommand{\phi}{\varphi}

\newcommand{\R}{\mathbb{R}}

\newcommand{\N}{\mathbb{N}}
\newcommand{\NN}{{{\cal N}^n}}

\newcommand{\eqdef}{\stackrel{\textrm{def}}{=}}

\newcommand{\opt}{\mathsf{opt}}

\newcommand{\C}{{{\cal C}}}
\newcommand{\D}{{{\cal D}}}
\newcommand{\X}{{\cal X}}

\title{Bounding the Sensitivity of Polynomial Threshold Functions}

\author{Prahladh Harsha  \qquad Adam Klivans \qquad Raghu Meka\\\\
{The University of Texas at Austin }\\
 {\small {\tt \{prahladh,klivans,raghu\}@cs.utexas.edu}}}
\date{}

\begin{document}

\begin{titlepage}

\maketitle
\thispagestyle{empty}

\begin{abstract}
We give the first nontrivial upper bounds on the average sensitivity and noise sensitivity of polynomial threshold functions.  More specifically, for a Boolean function $f$ on $n$ variables equal to the sign of a real, multivariate polynomial of total degree $d$ we prove 

\begin{itemize}
\item The average sensitivity of $f$ is at most $O(n^{1-1/(4d+6)})$ (we also give a combinatorial proof of the bound $O(n^{1-1/2^d})$. 

\item The noise sensitivity of $f$ with noise rate $\delta$ is at most $O(\delta^{1/(4d+6)})$.
\end{itemize}

Previously, only bounds for the degree $d = 1$ case were known
($O(\sqrt{n}$) and $O(\sqrt{\delta})$, for average and noise sensitivity
respectively). 

We highlight some applications of our results in learning theory where
our bounds immediately yield new agnostic learning algorithms and
resolve an open problem of Klivans et al.

The proof of our results use (i) the invariance principle of Mossel et
al., (ii) the anti-concentration properties of polynomials
in Gaussian space due to Carbery and Wright and (iii) new structural
theorems about random restrictions of polynomial threshold functions obtained via
hypercontractivity. 

These structural results may be of independent
interest, as they provide a generic template for transforming problems
related to polynomial threshold functions defined on the Boolean hypercube to
polynomial threshold functions defined in Gaussian space.
\end{abstract}

\end{titlepage}
 
\section{Introduction}
\subsection{Background}
Let $P$ be a real, multivariate polynomial of degree $d$, and let $f = \sign(P)$.  We say that the Boolean function $f$ is a polynomial threshold function (PTF) of degree $d$.  PTFs play an important role in computational complexity with applications in circuit complexity \cite{AspnesBFR1994,Beigel1993}, learning theory \cite{KlivansS2004,KlivansOS2004}, communication complexity \cite{Sherstov2008,Sherstov2009}, and quantum computing \cite{BealsBCMW2001}.  While many interesting properties (e.g., Fourier spectra, influence, sensitivity) have been characterized for the case $d=1$ of linear threshold functions (LTFs), very little is known for degrees $2$ and higher.  Gotsman and Linial \cite{GotsmanL1994} conjectured, for example, that the average sensitivity of a degree $d$ polynomial is $O(d \sqrt{n})$.  In this work, we take a step towards resolving this conjecture and give the first nontrivial bounds on the average sensitivity and noise sensitivity of degree $d$ PTFs (\newref[Theorem]{thm:as}) .

Average sensitivity \cite{BenOrL1985} and noise sensitivity \cite{KahnKL1988,BenjaminiKS1999} are two fundamental quantities that arise in the analysis of Boolean functions.   Roughly speaking, the average sensitivity of a Boolean function $f$ measures the expected number of bit positions that change the sign of $f$ for a randomly chosen input, and the noise sensitivity of $f$ measures the probability over a randomly chosen input $x$ that $f$ changes sign if each bit of $x$ is flipped independently with probability $\delta$ (we give formal definitions below). 

Bounds on the average and noise sensitivity of Boolean functions have direct applications in hardness of approximation \cite{Hastad2001,KhotKMO2007}, hardness amplification \cite{ODonnell2004}, circuit complexity \cite{LinialMN1993}, the theory of social choice \cite{Kalai2005}, and quantum complexity \cite{Shi2000}.  In this paper, we focus on applications in learning theory, where it is known that bounds on the noise sensitivity of a class of Boolean functions yield learning algorithms for the class that succeed in harsh noise models (i.e., work in the agnostic model of learning) \cite{KalaiKMS2008}.  We obtain the first efficient algorithms for agnostically learning PTFs with respect to the uniform distribution on the hypercube.  We also give efficient algorithms for agnostically learning ellipsoids in $\R^n$ with respect to the Gaussian distribution, resolving an open problem of Klivans et al. \cite{KlivansOS2008}. We discuss these learning theory applications in \newref[Section]{sec:learn}.

\subsection{Main Definitions and Results}

We begin by defining the (Boolean) noise sensitivity of a Boolean function:

\begin{definition}[Boolean noise sensitivity]
Let $f$ be a Boolean function $f:\dpmn \to \dpm$.  For any $\delta \in (0,1)$, let $X$ be a random element of the hypercube $\dpmn$ and $Z$ a $\delta$-perturbation of $X$ defined as follows: for each $i$ independently, $Z_i$ is set to $X_i$ with probability $1-\delta$ and $-X_i$ with probability $\delta$. The noise sensitivity of $f$, denoted $\NS_\delta(f)$, for noise $\delta$ is then defined as follows: $$\NS_\delta(f) = \Pr\left[f(X) \neq f(Z)\right].$$
\end{definition}

Intuitively, the Boolean noise sensitivity of $f$ measures the probability that $f$ changes value when a random input to $f$ is perturbed slightly.  In order to analyze Boolean noise sensitivity, we will also need to analyze the Gaussian noise sensitivity, which is defined similarly, but the random variables $X$ and $Z$ are drawn from a multivariate Gaussian distribution. Let $\calN=\calN(0,1)$ denote the univariate Gaussian distribution on $\R$ with mean $0$ and variance $1$.

\begin{definition}[Gaussian noise sensitivity]
Let $f: \R^n \to \{-1,1\}$ be any Boolean function on $\R^n$. Let $X,Y$ be two independent random variables drawn from the multivariate Gaussian distribution $\NN$ and $Z$ a $\delta$-perturbation of $X$ defined by $Z = (1-\delta) X + \sqrt{2\delta-\delta^2} Y$. The Gaussian noise sensitivity of $f$, denoted $\GNS_\delta(f)$, for noise $\delta$ is defined as follows: $$\GNS_\delta(f) = \Pr\left[f(X) \neq f(Z)\right].$$
\end{definition}

It is well known that the Boolean and Gaussian noise sensitivity of LTFs are at most $O(\sqrt{\delta})$. Our results give the first nontrivial bounds for degrees $2$ and higher in both the Gaussian and Boolean cases, with the Gaussian case being considerably easier to handle than the Boolean case.

\begin{theorem}[Boolean noise sensitivity]\label{thm:bns}
For any degree $d$ PTF $f:\{1,-1\}^n\to\{1,-1\}$ and $0 < \delta < 1$, $$\NS_\delta(f) = 2^{O(d)}\cdot \left(\delta^{1/(4d+6)}\right).$$
\end{theorem}

For the Gaussian case, we get a slightly better dependence on the degree d:
\begin{theorem}[Gaussian noise sensitivity]\label{thm:gns}
For any degree $d$ polynomial $P$ such that $P$ is either multilinear or corresponds to an ellipsoid, the following holds for the corresponding PTF $f=\sign(P)$. For all $0 < \delta < 1$, $$\GNS_\delta(f) = 2^{O(d)}\cdot \left(\delta^{1/(2d+1)}\right).$$
\end{theorem}

Diakonikolas et al.~\cite{DiakonikolasRST2009} prove that a similar bound holds for all degree $d$ PTFs. Our next set of results bound the average sensitivity or total influence of degree $d$ PTFs. 

\begin{definition}[average sensitivity]
Let $f$ be a Boolean function, and let $X$ be a random element of the hypercube $\dpmn$.  Let $X^{(i)} \in \dpmn$ be such that $X^{(i)}_i = -X_i$ and $X^{(i)}_j = X_j$ for $j \neq i$. Then, the influence of the $i^{th}$ variable is defined by $$\infl_i(f) = \Pr \left[ f\left(X\right) \neq f\left(X^{(i)}\right)\right].$$ The sum of all the influences is referred to as the average sensitivity of the function $f$, $$\AS(f) = \sum_i \infl_i(f).$$
\end{definition}

Clearly, for any function $f$, $\AS(f)$ is at most $n$. It is well known that the average sensitivity of ``unate'' functions (functions monotone in each coordinate), and thus of LTFs in particular is $O(\sqrt{n})$. This bound is tight as the Majority function has average sensitivity $\Theta(\sqrt{n})$. As mentioned before, Gotsman and Linial~\cite{GotsmanL1994} conjectured in 1994 that the average sensitivity of any degree $d$ PTF $f$ is $O(d \sqrt{n})$. We are not aware of any progress on this conjecture until now, with no $o(n)$ bounds known.

We give two upper bounds on the average sensitivity of degree d PTFs. We first use a simple translation lemma for bounding average sensitivity in terms of noise sensitivity of a Boolean function and \tref{bns} to obtain the following bound.
\begin{theorem}[average sensitivity]\label{thm:as}
  For a degree $d$ PTF $f:\{1,-1\}^n\to\{1,-1\}$, $$\AS(f) = 2^{O(d)}\cdot \left(n^{1-1/(4d+6)}\right).$$
\end{theorem}

We also give an elementary combinatorial argument, to show that the average sensitivity of any degree $d$ PTF is at most $3n^{1-1/2^d}$. The combinatorial proof is based on the following lemma for general Boolean functions that may prove useful elsewhere. For $x \in \dpmn$, and $i \in [n]$, let $x_{-i} = (x_1,\ldots,x_{i-1},x_{i+1},\ldots,x_n)$.
\begin{lemma}\label{mainlmc}
 For Boolean functions $f_i:\dpmn \rgta \dpm$ with $f_i$ not depending on the $i$'th coordinate $x_i$, and $X \in_u \dpm^n$, $$\ex_X\left[\, \left| \,\sum_i X_i f_i(X_{-i})\,\right|\,\right]^2 \leq 2 \sum_i \AS(f_i) + n.$$
\end{lemma}
We believe that when the functions $f_i$ in the above lemma are LTFs, the above bound can be improved to $O(n)$, which in turn would imply the Gotsman-Linal conjecture for quadratic threshold functions.
\subsection{Random Restritctions of PTFs -- a structural result} \label{sec:struct}
An important ingredient of our sensitivity bounds for PTFs are new structural theorems about random restrictions of PTFs obtained via {\em hypercontractivity}. The structural results we obtain can be seen as part of the high level ``randomness vs structure'' paradigm that has played a fundamental role in many recent breakthroughs in additive number theory and combinatorics. Specifically, we obtain the following structural result (Lemmas~\ref{case1} and \ref{case2}):  for any PTF, there exists a small set of variables such that with at least a constant probability, any random restriction of these variables satisfies one of the following: (1) the restricted polynomial is ``regular'' in the sense that no single variable has large influence or (2) the sign of the restricted polynomial is a very biased function.

We remark that our structural results, though motivated by similar results of Servedio~\cite{Servedio2007} and Diakonikolas et al.~\cite{DiakonikolasGJSV2009} for the simpler case of LTFs, do not follow from a generalization of their arguments for LTFs to PTFs. The structural results for random restrictions of low-degree PTFs provide a reasonably generic template for reducing problems involving arbitrary PTFs to ones on regular PTFs. In fact, these structural properties are used precisely for the above reason both in this work and in a parallel work by one of the authors, Meka and Zuckerman~\cite{MekaZ2009} to construct pseudorandom generators for PTFs.

\subsection{Related Work}
Independent of this work, Diakonikolas, Raghavendra, Servedio, and Tan \cite{DiakonikolasRST2009} have obtained nearly identical results to ours for both the average and noise sensitivity of PTFs.  The broad outline of their proof is also similar to ours.  In our proof, we first obtain bounds on noise sensitivity and then move to average sensitivity using a translation lemma. On the other hand, Diakonikolas et al.~\cite{DiakonikolasRST2009} first obtain bounds on the average sensitivity of PTFs and then use a generalization of Peres' argument~\cite{Peres2004}  for LTFs to move from average sensitivity to noise sensitivity. 
%Their observation yields a converse to our translation lemma for the specific case of PTFs.

Regarding our structural result described in \sref{struct}, Diakonikolas, Servedio, Tan and Wan \cite{DiakonikolasSTW2009} have independently obtained similar results to ours. As an application, they prove the existence of low-weight approximators for polynomial threshold functions.  
\subsection{Proof Outline}
The proofs of our theorems are inspired by the use of the {\sl invariance principle} in the proof of the ``Majority is Stablest'' theorem~\cite{MosselOO2005}. As in the proof of the ``Majority is Stablest'' theorem, our main technical tools are the invariance principle and the anti-concentration bounds (also called small ball probabilities) of Carbery and Wright~\cite{CarberyW2001}. 

Bounding the probability that a threshold function changes value either when it is perturbed slightly (in the case of noise sensitivity) or when a variable is flipped (average sensitivity) involves bounding probabilities of the form $\Pr\left[|Q(X)| \leq |R(X)|\right]$ where $Q(X), R(X)$ are low degree polynomials and $R$ has small $l_2$-norm relative to that of $Q$. The event $|Q(X)| \leq |R(X)|$ implies that either $|Q(X)|$ is small or $|R(X)|$ is large. In other words, for every $\gamma$ $$\Pr\left[|Q(X)| \leq |R(X)|\right] \leq \Pr\left[ |Q(X)| \leq \gamma\right] + \Pr\left[|R(X)| > \gamma\right].$$ Since $R$ has small norm, the second quantity in the above expression can be easily bounded using a tail bound (even Markov's inequality suffices). Bounding the first quantity is trickier. Our first observation is that if the random variable $X$ were distributed according to the Gaussian distribution as opposed to the uniform distribution on the hypercube, bounds on probabilities of the form $\Pr\left[|Q(X)| \leq \gamma\right]$ immediately follow from the anti-concentration bounds of Carbery and Wright~\cite{CarberyW2001}. We then transfer these bounds to the Boolean setting using the invariance principle.

Unfortunately, the invariance principle holds only for regular polynomials (i.e., polynomials in which no single variable has large influence). We thus obtain the required bounds on noise sensitivity and average sensitivity for the special case of regular PTFs. We then extend these results to an arbitrary PTF $f$ using our structural results on random restrictions of the PTF $f$. The structural results state that either the restricted PTF is a regular polynomial or is a very biased function. In the former case, we resort to the above argument for regular PTFs and bound the noise sensitivity of the given PTF. In the latter case, we merely note that the noise sensitivity of a biased function can be easily bounded. This in turn lets us extend the results for regular PTFs to all PTFs.

\section{Learning Theory Applications} \label{sec:learn}
In this section, we briefly elaborate on the learning theory applications of our results. 
%\subsection{Agnostic Learning}
Our bounds on Boolean and Gaussian noise sensitivity imply learning results in the challenging {\em agnostic} model of learning of Haussler~\cite{Haussler1992}  and Kearns, Schapire and Sellie~\cite{KearnsSS1994} which we define below. 

\begin{definition} \label{def:ag}
Let $\D$ be an arbitrary distribution on $\X$ and $\C$ a class of
Boolean functions $f: \X \to \{-1,1\}$.  For $\delta,\epsilon \in (0,1)$, we say that algorithm $A$
is a $(\delta,\epsilon)$-agnostic learning algorithm for ${\cal C}$ with respect to
$\D$ if the following holds. For any distribution $\D'$ on $\X\times \{-1,1\}$ whose marginal over $\X$ is $\D$, if
$A$ is given access to a set of labeled examples $(x,y)$ drawn from
$\D'$, then with probability at least $1 - \delta$ algorithm $A$
outputs a hypothesis $h: \X \to \{-1,1\}$ such that $$\Pr_{(x,y) \sim \D'}[h(x) \neq y] \leq \opt + \eps$$ where $\opt$ is the error made by the best classifier in $\C$, that is, $\opt = \inf_{g \in \C}\Pr_{(x,y)\sim \D'}[g(x) \neq y]$.
\end{definition}
\ignore{
In 2005, Kalai, Klivans, Mansour and Servedio showed the following sufficient condition for agnostically learning a concept class:

\begin{theorem}[{\cite{KalaiKMS2008}}] \label{thm:ag} Let $\C$ be a class of functions mapping
  $X$ to $\{-1,1\}$.  Let $D$ be a distribution on $\X \times \{-1,1\}$
  with marginal distribution $D_{X}$. Assume that for each $c \in \C$,
  there exists a polynomial $P$ of degree $d$ such that $\ex_{\D_{X}}\left[(P(x)
  - c(x))^2\right] \leq \epsilon^2$.  Then $\C$ is agnostically learnable to
  accuracy $\eps$ in time poly$(n^d/\epsilon)$.
\end{theorem}

Finally, the following theorem due to Klivans, O'Donnell and Servedio~\cite{KlivansOS2004} gives a precise
relationship between polynomial approximation and noise
sensitivity.
\begin{theorem}[{\cite{KlivansOS2004}}]
Let $f: \{1,-1\}^n \to \{1,-1\}$ with $\NS_{\delta}(f) \leq m(\delta)$ for
some (invertible) function $m:(0,1)\to(0,1)$.  Then there exists a constant $C$ and polynomial $P$ of degree at most $C \cdot
(1/m^{-1}(\delta))$ such that 

\[ \ex_{x \in_u \{1,-1\}^n} \left[\left(P(x) - f(x)\right)^2\right] \leq \delta. \]
\end{theorem}

We also have a similar theorem for Gaussian noise sensitivity (Klivans, O'Donnell and Servedio~\cite{KlivansOS2008}).  That is, we may replace $\NS_{\delta}(f)$ above with $\GNS_{\delta}(f)$ and the uniform distribution over
$\{-1,1\}^n$ with the multivariate Gaussian distribution $\NN$.
}
Kalai, Klivans, Mansour and Servedio~\cite{KalaiKMS2008} showed that the existence of low-degree real valued polynomial $l_2$-approximators to a class of functions, implies agnostic learning algorithms for the class. In an earlier result, Klivans, O'Donnell and Servedio~\cite{KlivansOS2004} gave a precise relationship between polynomial approximation and noise sensitivity, essentially showing that small noise sensitivity bounds imply good low-degree polynomial $l_2$-approximators. 

Combining these two results, it follows that bounding the noise sensitivity (either Boolean or Gaussian) of a concept class $\C$ yields an agnostic learning algorithm for $\C$ (with respect to the appropriate distribution). Thus, using our bounds on noise sensitivity of PTFs, we obtain corresponding learning algorithms for PTFs with respect to the uniform distribution over the hypercube.

\begin{theorem} \label{thm:mainlearn}
The concept class of degree $d$ PTFs is agnostically learnable to within $\epsilon$ with respect to the uniform distribution on $\{-1,1\}^n$ in time $n^{1/\epsilon^{O(d)}}$.
\end{theorem} 
\ignore{
\begin{theorem} \label{thm:mainlearn}
Let $C$ be the concept class of degree $d$ PTFs.   Then $C$ is agnostically learnable to within $\epsilon$ with respect to an arbitrary Gaussian on $\R^n$ or the uniform distribution on $\{-1,1\}^n$ in time $n^{1/\epsilon^{O(d)}}$.
\end{theorem}}
These are the first polynomial-time algorithms for agnostically learning constant degree PTFs with respect to the uniform distribution on the hypercube (to within any constant error parameter).  Previously, Klivans et al. \cite{KlivansOS2008} had shown that quadratic (degree $2$) PTFs corresponding to spheres are agnostically learnable with respect to spherical Gaussians on $\R^n$.  Our bounds on the Gaussian noise sensitivity of ellipsoids imply that this result can be extended to all ellipsoids with respect to (not necessarily spherical) Gaussian distributions thus resolving an open problem of Klivans et al.~\cite{KlivansOS2008}. 

It is implicit from a recent paper of Blais, O'Donnell and Wimmer~\cite{BlaisOW2008} that bounding the Boolean noise sensitivity for a concept class $\C$ yields non-trivial learning algorithms for a very broad class of discrete and continuous product distributions.  We believe this is additional motivation for obtaining bounds on a function's Boolean noise sensitivity.

\section{Organization}
The rest of the paper is organized as follows. We introduce the necessary notation and preliminaries in \newref[Section]{sec:notation}. We then present the structural results on random restrictions of PTFs (Lemmas~\ref{case1} and \ref{case2}) in \newref[Section]{sec:randrest}. In \sref{gns} we present our analysis of Gaussian noise sensitivity, followed by the analysis of Boolean noise sensitivity in \sref{bns}.  We remark that the analysis of the Gaussian noise sensitivity is simpler than the Boolean noise sensitivity analysis, since the Boolean case, in some sense, reduces to the ``regular'' or Gaussian case. We then present our bounds on average sensitivity of PTFs in \sref{as}.

\section{Notation and Preliminaries}\label{sec:notation}
We will consider functions/polynomials over $n$ variables $X_1,\dots,X_n$. Corresponding to any set $I \subseteq [n]$ (possibly multi-set), there is a monomial $X^{I}$ defined as $X^I = \prod_{i\in I}X_i$. The degree of the monomial $X^I$ is the size of the set $I$, denoted by $|I|$. Note that if $I$ is a ``regular'' set (opposed to a multi-set), then the monomial $X^I$ is linear in each of the participating variables $X_i, i\in I$. 

A polynomial of degree $d$ is a linear combination of monomials of degree at most $d$, that is, $P(X_1,\ldots,X_n) = \sum_{I\subseteq [n],|I| \leq d} a_I X^I$. The $a_I$'s are called the coefficients of the polynomial $P$. By convention, we set $a_I =0$ for all other $I$. If the above summation is only over sets $I$ and not multi-sets, then the polynomial is said to be {\em multilinear}. Observe that while working over the hypercube, it suffices to consider only multilinear polynomials. We use the following notations throughout.
\begin{enumerate}
\item Unless otherwise stated, we work with a PTF $f$ of degree $d$ and a degree $d$ polynomial $P(X) = \sum_I a_IX^I$ with zero constant term (i.e., $a_\emptyset = 0$) such that  $f(X_1,\ldots,X_n) = \sign(P(X_1,\ldots,X_n)-\theta)$. In case of ambiguity, we will refer to the coefficients $a_I$ as $a_I(P)$. 
\item For a polynomial $P$ as above and an underlying distribution over $X=(X_1,\dots,X_n)$, the $l_2$-norm of the polynomial over $X$ is defined by $\|P\|^2 = \ex\left[P(X)^2\right]$. Note that if $P$ is a multilinear polynomial and the distribution is either the multivariate Gaussian $\NN$ or the uniform distribution over the hypercube, then $\|P\|^2 = \sum_I a_I^2$.
\item  For $i \in [n]$, $x^i = (x_1,\ldots,x_i) \in \dpm^i$, $f_{x^i}: \dpm^{n-i} \rgta \dpm$ is defined by $f_{x^i}(X_{i+1},\ldots,X_n) = \sign(P(x_1,\ldots,x_i,X_{i+1},\ldots,X_n) - \theta)$.
\item For $i \in [n]$, $P_{|i}(X_1,\ldots,X_i) = \sum_{I \subseteq [i]} a_I X^I$ is the {\em restriction} of $P$ to the variables $X_1,\ldots,X_i$.
\item For a multi-set $S$, $x \in_u S$ denotes an uniformly chosen element from $S$.
\item For clarity, we supress the exact dependence of the constants on the degree $d$ in this extended abstract; a more careful examination of our proofs shows that all constants depending on the degree $d$ are at worst $2^{O(d)}$.
\end{enumerate}

\begin{definition}
A partial assignment $x^i = (x_1,\ldots,x_i)$ is $\epsilon$-{\em determining} for $f$, if there exists $b \in \{1,-1\}$ such that $\pr_{(X_{i+1},\ldots,X_n)\in_u \dpm^{n-i}}[\, f_{x^i}(X_{i+1},\ldots,X_n) \neq b\,] \leq \epsilon$.
\end{definition}

We now define {\em regular polynomials} which play an important role in all our results. Intuitively, a polynomial is regular if no variable has high influence. For a polynomial $Q$, the {\em weight} of the $i^{th}$ coordinate is defined by $w^2_i(Q) = \sum_{I \ni i} a_I^2$. For $i \in [n]$, let $\sigma_i(Q)^2 = \sum_{j \geq i} w_j^2(Q)$. 
\begin{definition}[regular polynomials]
A multilinear polynomial $Q$ is {\em $\epsilon$-regular} if $\sum_i w_i^4(P)\leq \epsilon^2 \left(\sum_i w^2_i(P)\right)^2  = \epsilon^2\sigma^4_{1}(P)$. A PTF $f(x) = \sign(Q(x)-\theta)$ is $\epsilon$-regular if $Q$ is $\epsilon$-regular.
\end{definition}
We also assume without loss of generality that the variables are ordered such that $w_1(P) \geq w_2(P) \geq \cdots \geq w_n(P)$. 

We repeatedly use three powerful tools: $(2,4)$-hypercontractivity, the invariance principle of Mossel et al.~\cite{MosselOO2005} and the anti-concentration bounds of Carbery and Wright~\cite{CarberyW2001}. We state the relevant results below.

\begin{lemma}[$(2,4)$-hypercontractivity]\label{hypercon}
If $Q,R$ are degree $d$ multilinear polynomials, then for $X \in_u \dpm^n$, $\ex_X\,[Q^2\cdot R^2]\,] \leq 9^d \cdot \ex_X[Q^2]\cdot \ex_X[R^2]$. In particular, $\ex[Q^4] \leq 9^d \cdot \ex[Q^2]^2$.
\end{lemma}

The following anti-concentration bound  is a special case of Theorem~8 of \cite{CarberyW2001} (in their notation, set $q = 2d$ and the log-concave distribution $\mu$ to be $\NN$).
\begin{theorem}[Carbery-Wright anti-concentration bound]\label{cw}
There exists an absolute constant $C$ such that for any polynomial $Q$ of degree at most $d$ with $\|Q\| = 1$ and any interval $I \subseteq \reals$ of length $\alpha$, $\pr_{X\leftarrow \NN}[ Q(X )\in I] \leq C d \,\alpha^{1/d}$.
\end{theorem}

\newcommand{\psif}{\|\psi^{(4)}\|_\infty}

The following result due to Mossel et al.~\cite{MosselOO2005} generalizes the classical quantitative central limit theorem for sums of independent variables, Berry-Ess\'een Theorem, to low-degree polynomials over independent variables.
\begin{theorem}[Mossel et al.]\label{invariance}
There exists a universal constant $C$ such that the following holds. For any $\epsilon$-regular multilinear polynomial $P$ of degree at most $d$ with $\|P\| = 1$ and $ t \in \reals$, 
\[ \left|\,\pr_{X \in_u \{1,-1\}^n}[ P(X) < t] - \pr_{Y \lfta \NN}[P(Y) < t]\,\right| \leq C^d \eps^{2/(4d+1)}.\]
%\[ \left|\ex_{X\in_u \{1,-1\}^n}\left[\psi(P(X))\right] - \ex_{Y \gets \NN}\left[\psi(P(Y))\right]\right| \leq d\, 9^d\,\psif\,\left(\sum_i w_i^4(P)\right). \]
\end{theorem}
The result stated in \cite{MosselOO2005} uses $\max_i w^2_i(P)$ as the notion of regularity instead of $\sum_i w_i^4(P)$ as we do. However, their proof extends straightforwardly to the above.

\section{Random Restrictions of PTFs}\label{sec:randrest}
We now establish our structural results on random restrictions of low-degree PTFs. The use of {\em critical indices} ($K(P,\epsilon)$) in our analysis is motivated by the results of Servedio \cite{Servedio2007} and Diakonikolas et al.~\cite{DiakonikolasGJSV2009} who obtain similar results for LTFs. At a high level, we show the following.

Given any $\epsilon>0$, define the $\epsilon$-critical index of a multilinear polynomial $P$,  $K = K(P,\epsilon)$, to be the least index $i$ such that $w^2_j(P) \leq \epsilon^2\, \sigma^2_{i+1}(P)$ for all $j > i$. We consider two cases depending on how large $K(P,\epsilon)$ is and roughly, show the following (here $c, \alpha > 0$ are some universal constants).
\begin{enumerate}
\item $K \leq 1/\epsilon^{cd}$. In this case we show that for $x^K = (x_1,\ldots,x_K) \in_u \dpm^K$, the PTF $f_{x^K}$ is $\epsilon$-regular with probability at least $\alpha$. %We then use the results of Carbery and Wright and the invariance principle of Mossel et al.~to bound the noise-sensitivity of $\epsilon$-regular PTFs.
\item $K > 1/\epsilon^{cd}$. In this case we show that with probability at least $\alpha$, the value of the threshold function is determined by the {\em top} $L = 1/\epsilon^{cd}$ variables. %Specifically, we show that $x^L = (x_1,\ldots,x_L) \in_u \dpm^L$, with probability at least $\alpha$, the PTF $f_{x^L}$ is an almost-constant function.
\end{enumerate}

More concretely, we show the following. 
\begin{lemma}\label{case1}
For every integer $d$, there exist constants $a_d\in \reals$, $\gamma_d > 0$ such that for any multilinear polynomial $P$ of degree at most $d$ and $K=K(P,\epsilon)$ as defined above, the following holds.
The polynomial $P_{x^K}(Y_{k+1},\ldots,Y_n) \eqdef P(x_1,\ldots,x_K,Y_{K+1},\ldots,Y_n)$ in variables $Y_{K+1},\ldots,Y_n$ obtained by randomly choosing $x^K=(x_1,\ldots,x_K) \in_u \dpm^K$ is $a_d \epsilon$-regular with probability at least $\gamma_d$.
\end{lemma}

\begin{lemma}\label{case2}
For every $d$, there exist constants $b_d,c_d \in \reals$, $\delta_d > 0$, such that for any multilinear polynomial $P$ of degree at most $d$ the following holds. If $K(P,\epsilon) \geq c_d\log(1/\epsilon)/\epsilon^2 = L$, then a random partial assignment $(x_1,\ldots,x_L) \in_u \dpm^L$ is $b_d \epsilon$-determining for $P$ with probability at least $\delta_d$.
%For any $\theta \in \reals$,
%\[ \pr_{x \in_u \dpm^n}[ \sgn(P(x) - \theta) \neq \sgn(P_{|L}(x_1,\ldots,x_L) - \theta)] \leq \delta_d\]
%for a universal constant $\delta_d < 1$.
\end{lemma}

To prove the above structural properties we need the following simple lemmas. 
\begin{lemma}[{\cite[Lemma~3.2]{AlonGK2004}}]\label{problm}
Let $A$ be a real valued random variable satisfying $\ex[A] = 0$, $\ex[A^2] = \sigma^2$ and $\ex[A^4] \leq b \sigma^4$. Then, $\pr[\,A \geq \sigma/4\sqrt{b}\,] \geq 1/4^{4/3}b$.
\end{lemma}

\begin{lemma}\label{polylb1}
For $d > 0$ there exist constants $\alpha_d, \beta_d > 0$ such that for any degree at most $d$ polynomial $Q$, and $X \in_u \dpm^n$, $\pr[\,Q(X) \geq  \ex[Q] + \alpha_d \sigma(Q)\,] \geq\beta_d$,
where $\sigma^2(Q)$ is the variance of $Q(X) = \|Q\|^2-\left(\ex_X[Q]\right)^2$. In particular, $\pr[\,Q(X) \geq  \ex[Q] \,] \geq\beta_d$.
\end{lemma}

\begin{proof}
Let random variable $A = Q(X)-\ex_X[Q(X)]$. Then, $\ex[A] = 0$, $\ex[A^2] = \sigma^2(Q)$ and by $(2,4)$-hypercontractivity, $\ex[A^4] \leq 9^d\, \ex[A^2] = 9^d \sigma^4(Q)$. The claim now follows from \newref[Lemma]{problm}.
\end{proof}

\subsection{Proof of {\newref[Lemma]{case1}}}
\begin{proof}
Let $X \equiv (X_1,\ldots,X_K)$. We prove the lemma as follows: (1) Bound the expectation of $\sum_{j > K} w^4_j(P_X)$ using hypercontractivity and use Markov's inequality to show that with high probability $\sum_{j > K} w_j^4(P_X)$ is small. (2) Use the fact that $\sigma_{K+1}^2(P_X)= \sum_{j > K} w_j^2(P_X)$ is a degree at most $2d$ polynomial in $X$ and \newref[Lemma]{polylb1} to lower bound the probability that $\sigma_{K+1}^2(P_X)$ is large. Let 
\begin{multline*}
   P_X(Y_{K+1},\ldots,Y_n) = P(X_1,\ldots,X_K,Y_{K+1},\ldots,Y_n) = \\R(X_1,\ldots,X_K) + \sum_{J \subseteq [K+1,n], 0 < |J| \leq d} Q_J(X_1,\ldots,X_K)\, \prod_{j \in J} Y_j.
\end{multline*}

We now bound $\ex[\sum_{j > K} w_j^4(P_X)]$. Fix a $j  > K$ and observe that $w_j^2(P_X) = \sum_{J \ni j} Q_J^2(X)$. Thus, 
\begin{equation}\label{eq0}
\ex_X\left[w_j^2(P_X)\right] = \sum_{J \ni j} \ex_X\left[Q_J^2(X)\right] = \sum_{J \ni j} \|Q_J\|^2 = w^2_j(P).  
\end{equation}
 Further, by $(2,4)$-hypercontractivity, \newref[Lemma]{hypercon},
\begin{multline*}\label{}
  \ex_X\left[w_j^4(P_X)\right] = \ex_X\left[\,\sum_{J_1,J_2 \ni j}\, Q_{J_1}^2(X) \,Q_{J_2}^2(X)\,\right]  = \sum_{J_1,J_2 \ni j}\, \ex_X\left[\,Q_{J_1}^2(X) \,Q_{J_2}^2(X)\,\right]\\
\leq  \sum_{J_1,J_2 \ni j} 9^d \ex_X\left[ Q_{J_1}^2(X)\right]\cdot \ex_X\left[ Q_{J_2}^2(X)\right] = \sum_{J_1,J_2 \ni j}\, 9^d \,\|Q_{J_1}\|^2\,\|Q_{J_2}\|^2 = 9^d \, w^4_j(P).
\end{multline*}
 Hence, $\ex[\, \sum_{j > K} w_j^4(P_X)\,] \leq 9^d \sum_{j>K} w^4_j(P)$. Now, from the definition of $K(P,\epsilon)$, $w^2_j(P) \leq \epsilon^2 \sigma^2_{K+1}(P)$ for all $j > K$. Thus, $$\sum_{j > K} w^4_j(P) \leq \epsilon^2 \sigma^2_{K+1}(P) \sum_{j > K} w^2_j(P) = \epsilon^2 \sigma^4_{K+1}(P).$$
Combining the above inequalities and applying Markov's inequality we get
\begin{equation}\label{eq1}
  \pr_X\,[\, \sum_{j > K} w_j^4(P_X) \geq \gamma 9^d \epsilon^2 \sigma^4_{K+1}(P)\,] \leq 1/\gamma .
\end{equation}
Observe that $Q(X) = \sum_{j > K} w_j^2(P_X)$ is a degree at most $2d$ polynomial in $X_1,\ldots,X_k$ and by \eqref{eq0}, $\ex\left[Q\right] = \sum_{j> K} w^2_j(P) = \sigma^2_{K+1}(P)$. Thus, by applying \newref[Lemma]{polylb1} to $Q$, $\pr\,[\,\sum_{j > K} w_j^2(P_X) \geq \sigma^2_{K+1}(P)\,] \geq \beta_{2d}$. Setting $\gamma = 2/\beta_{2d}$ in  \eqref{eq1} and using  the above equation, we get 
\[ \pr_X\left[\, \sum_{j > K} w_j^4(P_X) \,\leq\, a_d^2 \epsilon^2\, \left(\sum_{j > K} w_j^2(P_X)\right)^2\right] \geq \beta_{2d}/2,\]
where $a_d^2 = 2 \cdot 9^d/\beta_{2d}$. Thus, the polynomial $P_X(Y_{K+1},\ldots,Y_n)$ is $(a_d \epsilon)$-regular with probability at least $\gamma_d = \beta_{2d}/2$.
\end{proof}

\subsection{Proof of {\newref[Lemma]{case2}}}
We use the follwing simple lemma.
\begin{lemma}\label{simple1}
For $1 \leq i < j < K(P,\epsilon)$, $\sigma^2_j(P) \leq (1-\epsilon^2)^{j-i}\sigma^2_i(P)$.
\end{lemma}
\begin{proof}
For $1 \leq i < K(P,\epsilon)$, we have 
\[ \sigma^2_i(P) = w^2_i(P) + \sigma^2_{i+1}(P) \geq \epsilon^2 \sigma^2_i(P) + \sigma^2_{i+1}(P).\]
Thus, $\sigma^2_{i+1}(P) \leq (1-\epsilon^2) \sigma^2_i(P)$. The lemma follows.
\end{proof}

\begin{proof}[Proof of {\newref[Lemma]{case2}}]
Suppose that $K(P,\epsilon) \geq L = c\log(1/\epsilon)/\epsilon^2$ for a constant $c$ to be chosen later and let $Q(X_1,\ldots,X_n) = P(X_1,\ldots,X_n) - P_{|L}(X_1,\ldots,X_L)$. The proof proceeds as follows. We first show that $\|Q\|$ is significantly smaller than $\|P_{|L}\|$. We then use \newref[Lemma]{polylb1} applied to $P_{|L}-\theta$ and Markov's inequality applied to $|Q(X)|$ to show that $|P_{|L}(X_1,\ldots,X_L)-\theta|$ is larger than $|Q(X)|$, so that $Q(X)$ cannot flip the sign of $P_{|L}(X_1,\ldots,X_L)-\theta$, with at least a constant probability. We first bound $\|Q\|$. %in \newref[Claim]{case2cl}, whose proof is defered to \aref{randrest}. 
\begin{claim}\label{case2cl}
For a suitably large enough constant $c_d$, $\|Q\|\leq \sqrt{\epsilon}\,\alpha_d \, \|P_{|L}\|$.
\end{claim}
\begin{proof}
Let $\alpha_d, \beta_d$ be the constants from \newref[Lemma]{polylb1}. By definition $\|Q\|^2 = \sum_{I: I \not\subseteq [L]} a_I^2 \leq \sigma^2_L(P)$. Now,
\begin{align*}\label{}
\sigma^2_1(P) &= \sum_{j < L} w^2_j(P) + \sigma^2_L(P)\\
 &\leq d \sum_{I: I \cap [L] \neq \emptyset} a_I^2 + \sigma^2_L(P)\\
 &\leq d \sum_{I:\emptyset \neq I \subseteq [L]} a_I^2 + d \sum_{I: I \not\subseteq [L]} a_I^2 + \sigma^2_L(P)\\
 &\leq d \sum_{I: \emptyset\neq I \subseteq [L]} a_I^2 + d \sum_{j > L} w^2_j(P) + \sigma^2_L(P)\\
 &\leq d \sum_{I: \emptyset \neq I \subseteq [L]} a_I^2 + (d+1)\, \sigma^2_L(P).
\end{align*}
Further, by \newref[Lemma]{simple1}, $\sigma^2_L(P) \leq (1-\epsilon^2)^{L-1} \sigma^2_1(P) $.
Combining the above inequalities we get,
\begin{equation}\label{eq2-2}
 \sigma^2_L(P) \leq O_d\left(\,(1-\epsilon^2)^{L-1}\,\right) \sum_{I:\emptyset\neq  I \subseteq [L]} a_I^2 = O_d\left(\,(1-\epsilon^2)^{L-1}\,\right) \,\sigma^2(P).  
\end{equation}
Choosing $L = c_d \log(1/\epsilon)/\epsilon^2$ for large enough $c_d$, we get the claim.
\end{proof}

%Let $D$ be a $2d$-wise independent distribution over $\dpm^{n-L}$. Let $D'$ be the product distribution $\mathcal{U}_L \times D$ over $\dpm^n$, where $\mathcal{U}_L$ denotes the uniform distribution over $\dpm^L$. Then, since $D'$ is $2d$-wise independent,
%\[ \ex_{x \leftarrow D'}\,[\,Q(x_1,\ldots,x_n)^2\,] = \|Q\|^2.\]
By \newref[Claim]{case2cl} and Markov's inequality,
\begin{equation}\label{eqcase2-1}
 \pr_{x \in_u \dpm^n} \left[\,\left|Q(x_1,\ldots,x_n)\right| \geq \alpha_d \,\|P_{|L}\| \,\right] \leq \pr_{x \in_u \dpm^n} \left[\,\left|Q(x_1,\ldots,x_n)\right| \geq \|Q\|/\sqrt{\epsilon}\,\right] \leq \epsilon.
\end{equation}
Let $S \subseteq \dpm^L$ be the set of all {\em bad} $x^L \in \dpm^L$ such that,
\[ \pr_{(X_{L+1},\ldots,X_n) \in_u \dpm^n}\left[\,\left|Q(x_1,\ldots,x_L,X_{L+1},\ldots,X_n)\right| \geq \alpha_d \,\|P_{|L}\|\,\right] \geq 2 \epsilon/\beta_d.\]
Then, from \eqref{eqcase2-1} and the above equation, $\pr_{x^L \in_u \dpm^L}\left[ x^L \in S\right] \leq \beta_d/2$.
 Now, let $T \subseteq \dpm^L$ be such that for $x^L \in T$, $|\,P_{|L}(x_1,\ldots,x_L) - \theta\,| \geq \alpha_d\, \|P_{L}\|$ and $x^L \notin S$. Observe that all $x^L \in T$ are $(2\epsilon/\beta_d)$-determining and by \newref[Lemma]{polylb1} and the above equations,
\[ \pr_{x^L \in_u \dpm^L}\left[ x^L \in T\right] \geq \pr_{x^L \in_u \dpm^L}[\left|P_{|L}(x_1,\ldots,x_L) - \theta\right| \geq \alpha_d\, \|P_{L}\|] - \pr_{x^L \in_u \dpm^L}\left[ x^L \in S\right] \geq \beta_d/2.\]
The lemma now follows.
\end{proof}

\section{Gaussian Noise Sensitivity of PTFs} \label{sec:gns}
%In this section, we bound the Gaussian noise sensitivity of PTFs and thus prove \newref[Theorem]{thm:gns}.  We consider this a ``warm-up'' for the case of Boolean noise sensitivity, since the Boolean case, in some sense, reduces to the ``regular'' or Gaussian case.  Also, the proof is considerably simpler and only makes use of an anti-concentration bound for polynomials in Gaussian space
In this section, we bound the Gaussian noise sensitivity of PTFs and thus prove \newref[Theorem]{thm:gns}. The proof is simpler than the Boolean case and only makes use of an anti-concentration bound for polynomials in Gaussian space.

Although \newref[Theorem]{thm:gns} was stated only for multilinear polynomials and ellipsoids, we give a proof below that works for all degree $d$ polynomials using ideas from Diakonikolas et al.~\cite{DiakonikolasRST2009}, who were the first to prove a bound on the Gaussian noise sensitivity of general degree $d$ polynomials (see remarks after the statement of \newref[Claim]{claim:smallQ}). 

\begin{proof}[Proof of {\newref[Theorem]{thm:gns}}]
Let $f$ be the degree $d$ PTF and $P$ the corresponding degree $d$ polynomial such that $f(x) = \sign(P(x))$. We may assume without loss of generality. that $P$ is normalized, i.e., $\|P\|^2 = \ex[P^2(X)] =1$. 

The proof is based on the Carbery-Wright anti-concentration bound (\newref[Theorem]{cw}) for degree $d$ PTFs. Let $X,Y \sim \NN$ and $Z \eqdef (1-\delta)\,X + \sqrt{1-(1-\delta)^2}\,Y =(1-\delta)\,X + \sqrt{2\delta - \delta^2}\,Y$. Let $\rho = \sqrt{2\delta - \delta^2}$. Define the perturbation polynomial $Q(X,Y) = P(Z) - P(X) = P((1-\delta)X+\rho Y) - P(X)$.

Now, for $\gamma > 0$ to be chosen later,
\begin{align*}
  \pr[\sgn(P(X)) \neq \sgn(P(Z))] &= \pr[\sgn(P(X)) \neq \sgn(P(X) + Q(X,Y))]\\
  &= \pr[|P(X)| < |Q(X,Y)|]\\
  &\leq \pr[|P(X)| < \gamma] + \pr[|Q(X,Y) > \gamma]\\
  &\leq C_d \gamma^{1/d} + \pr[|Q(X,Y)| > \gamma],
\end{align*}
where the last inequality follows from the anti-concentration bound (\newref[Theorem]{cw}). 
 In \newref[Claim]{claim:smallQ}, we show that the norm $\|Q\|$ of the pertubation polynomial is at most $c_d\sqrt{\delta}$ for some constant $c_d$ (dependant on $d$).   We can now apply Markov's inequality to bound the second quantity as follows.
\[\pr[|Q(X,Y)| > \gamma] \leq \|Q\|^2 /\gamma^2 \leq c_d\, \delta/\gamma^2.\]
Thus, 
\[ \GNS_\delta(f) \leq C_d \gamma^{1/d} + \frac{c_d \delta}{\gamma^2}.\]
The theorem follows by setting $\gamma = \delta^{d/(2d+1)}$ in which case we get $\GNS_\delta(f) = O_d(\delta^{1/(2d+1)})$.
\end{proof}

We note that we can get a slightly stronger bound of $O_d\left(\delta^{1/2d}\sqrt{\log(1/\delta)}\right)$ if we used a stronger tail bound instead of Markov's in the above argument.

\begin{claim}\label{claim:smallQ}There exists a constant $c_d$ such that $\|Q\| \leq c_d\sqrt{\delta}$.
\end{claim}
An earlier version of this paper had an error in the proof of this claim. As pointed out to us by the authors of \cite{DiakonikolasRST2009}, that proof worked only for multilinear polynomials and ellipsoids. Diakonikolas~et al.~\cite{DiakonikolasRST2009} proved the claim for general degree $d$ polynomials. For the sake of completeness, we give a simplified presentation of their proof (that works for all degree $d$ polynomials) in \newref[Section]{sec:claimproof}.

\section{Noise sensitivity of PTFs}\label{sec:bns}
We now bound the noise sensitivity of PTFs and prove \newref[Theorem]{thm:bns}. We do so by first bounding the noise sensitivity of regular PTFs and then use the results of the previous section to reduce the general case to the regular case.

\subsection{Noise sensitivity of Regular PTFs}
At a high level, we bound the noise sensitivity of regular PTFs as follows: (1) Reduce the problem to that of proving certain anti-concentration bounds for regular PTFs over the hypercube. (2) Use the invariance principle of Mossel et al.~\cite{MosselOO2005} to reduce proving anti-concentration bounds over the hypercube to that of proving anti-concentration bounds over Gaussian distributions. (3) Use the Carbery-Wright anti-concentration bounds~\cite{CarberyW2001} for polynomials over log-concave distributions.% This step is similar to the Gaussian noise sensitivity analysis.

For the rest of this section, we fix degree $d$ multilinear polynomial $P$ and a corresponding degree $d$ PTF $f$. Recall that it suffices to consider multilinear polynomials as we are working over the hypercube. We first reduce bounding noise sensitivity to proving anti-concentration bounds. 
\begin{lemma}\label{actons}
For $ 0 < \rho < 1$, $\delta > 0$, $$\NS_\rho(f) \leq (d +1)\,\delta + \pr_{x \in \dpm^n}[\,|P(x)-\theta| \leq 2 \sqrt{\rho}/\delta\,].$$
\end{lemma}
\begin{proof}
Let $S$ be a random subset $S \subseteq [n]$ where each $i \in [n]$ is in $S$ independently with probability $\rho$. From the definition of noise sensitivity it easily follows that 
\begin{align}\label{refns-0}
\NS_\rho(f) & = \pr_{X \in_u \dpm^n,S}\,[\, \sign\left(P(X) -\theta\right) \neq \sign(\,P(X) - 2\sum_{I:|I\cap S|\text{ is odd}} a_I X^I-\theta\,)\,]\nonumber\\
& = \pr_{X \in_u \dpm^n,S}[\,\left|P(x)-\theta\right| \leq 2\, |\,\sum_{I:|I\cap S|\text{ is odd}} a_I X^I\,|\,]\nonumber\\
& \leq \pr_{X \in_u\dpm^n, S}\,[\,|\,\sum_{I:|I\cap S|\text{ is odd}} a_I X^I\,|\,\geq \sqrt{\rho}/\delta\,] + \pr_{X \in_u \dpm^n}\,[\,|\, P(X) - \theta\,|\, \leq 2\sqrt{\rho}/\delta\,]
\end{align}
Define a non-negative random variable $P_S$ as follows: $P_S^2 = \sum_{I: |I \cap S|\text{ is odd}} a_I^2$. We can then bound the first quantity in the above expression using $P_S$ as follows:
\begin{multline}\label{refns-1}
 \pr_{X \in_u\dpm^n, S}\,[\,|\,\sum_{I:|I\cap S|\text{ is odd}} a_I X^I\,|\,\geq \sqrt{\rho}/\delta\,]  \leq \pr_{X \in_u\dpm^n, S}\,[\,|\,\sum_{I:|I\cap S|\text{ is odd}} a_I X^I\,|\,\geq P_S/\sqrt{\delta}\,]\\ + \pr_{S}\,[\,P_S \geq \sqrt{\rho}/\sqrt{\delta}\,] 
\end{multline}
Since $\ex_X (\,\sum_{I |I\cap S| \text{ is odd}}a_I X^I\,)^2=P_S^2$, by Markov's inequality, we have
\begin{equation}\label{refns-2}
 \pr_{x \in_u \dpm^n}\,[\, \,|\,\sum_{I: |I \cap S|\text{ is odd}} a_I X^I\,|\, \geq P_S/\sqrt{\delta}\,] \leq \delta.
\end{equation}
Now, note that $P_S^2 \leq \sum_{i \in S} w_i^2(P)$. Thus, $\ex_S[P^2_S] \leq \ex_S[\,\sum_{i\in S} w_i^2(P)\,] = \rho \sum_i w_i^2(P) \leq d \, \rho$. Hence, by Markov's inequality, $\pr_S[\,P_S \geq \sqrt{\rho}/\sqrt{\delta}\,] \leq d \,\delta$.
The lemma now follows by combining Equations \eqref{refns-0}, \eqref{refns-1}, \eqref{refns-2} and the above equation.
\end{proof}

We now prove an anti-concentration bound for regular PTFs.

\begin{lemma}\label{regac}
If $P$ is $\epsilon$-regular, then for any interval $I\subseteq \reals$ of length at most $\alpha$, $$\pr_{X \in_u \dpm^n}[\,P(X) \in I\,] = O_d(\, \alpha^{1/d} + \epsilon^{2/(4d+1)}\,).$$
\end{lemma}
\begin{proof}
Let $Z_1= P(X), Z_2= P(Y)$ for $X \in_u \dpm^n, Y \leftarrow \NN$. Then, since $P$ is $\epsilon$-regular, by \newref[Theorem]{invariance}, for all $t \in \reals$, $|\,\pr[Z_1 > t]  - \pr[Z_2>t]\,| = O_d(\epsilon^{2/(4d+1)})$. Now, by the above equation and \newref[Theorem]{cw} applied to the random variable $Y$ for interval $I$, $\pr[Z_1 \in I] = \pr[Z_2\in I] + O_d(\,\epsilon^{2/(4d+1)}\,) = O_d(\, \alpha^{1/d} + \epsilon^{2/(4d+1)}\,)$.
\end{proof}

We can now obtain a bound on noise sensitivity of regular PTFs. 
\begin{theorem}\label{nsregular}
If $f$ is an $\epsilon$-regular PTF of degree $d$, then $\NS_{\epsilon}(f) \leq O_d\left(\,\epsilon^{1/(2d+2)}\right)$.
\end{theorem}
\begin{proof}
Let $\delta > 0$ to be chosen later. Then, by \newref[Lemma]{actons} and \newref[Lemma]{regac} above, $\NS_\epsilon(f) = O_d(\, \delta + \epsilon^{2/(4d+1)} + \epsilon^{1/2d}/\delta^{1/d}\,)$.
Choosing $\delta = \epsilon^{1/(2d+2)}$ we get $\NS_\epsilon(f) = O_d(\, \epsilon^{1/(2d+2)}\,)$.
\end{proof}

\subsection{Noise Sensitivity of arbitrary PTFs}\label{nsarbit}
We prove \newref[Theorem]{thm:bns} by recursively applying the following lemma. 
\begin{lemma}\label{mainlmns}
For every $d$ there exist universal constants $c_d,\Delta_d \in \N,\alpha_d \in(0,1)$ such that for $M = \min(K(P,\epsilon), c_d\log(1/\epsilon)/\epsilon^{2})$ and $X^M = (X_1,\ldots,X_M)\in_u \dpm^M$,
\begin{equation}\label{maineq}
  \pr_{X^M}\,\left[\, \NS_\epsilon(f_{X^M}) \leq \Delta_d \epsilon^{1/(2d+2)}\,\right] \geq \alpha_d.
\end{equation}
\end{lemma}
\begin{proof}
Let $a_d,b_d,c_d, \gamma_d,\delta_d$ be the constants from Lemmas~\ref{case1},~\ref{case2}. Let $\alpha_d = \min(\gamma_d,\delta_d)$. We consider two cases.

\noindent {Case (i): }$M = K(P,\epsilon)$. Then, by \newref[Lemma]{case1} and \newref[Theorem]{nsregular}, for $X^K \in_u \dpm^K$, with probability at least $\alpha_d$, $\NS_\epsilon(f_{x^K}) \leq \Delta_d \epsilon^{1/(2d+2)}$ for some constant $\Delta_d$.

\noindent {Case (ii):} $M =   c_d\log(1/\epsilon)/\epsilon^{2}$. Then, by \newref[Lemma]{case2}, $X^M \in_u \dpm^M$ is $b_d\epsilon$-determining with probability at least $\alpha_d$. Further, if $X^M$ is $b_d\epsilon$-determining, with $f_{X^M}$ biased towards $b \in \dpm$, then 
\begin{multline*}
  \NS_\epsilon(f_{X^M}) =  \pr_{Z_1 \in_u \dpm^{n-M}, Z_2 \in_\epsilon Z_1}\,[ \,f_{X^M}(Z_1) \neq f_{X^M}(Z_2)\,] \leq 2\, \pr_{Z \in_u \dpm^{n-M}}\,[\,f_{X^M}(Z) \neq b\,] \leq 2b_d\epsilon,
\end{multline*}
where $Z_2 \in_\epsilon Z_1$ is an $\epsilon$-perturbation of $Z_1$. The lemma now follows.
\end{proof}

\begin{proof}[Proof of {\newref[Theorem]{thm:bns}}]
Let $c_d,\Delta_d,\alpha_d$ be as in the above lemma and let $L = c_d \log(1/\epsilon)/\epsilon^2$, $t = \log_{1-\alpha_d}(1/\epsilon)$. We will show that for $\delta = \epsilon^{1/(2d+2)}/(L\, t)= O_d(\epsilon^{(4d+5)/(2d+2)}/\log^2(1/\epsilon))$,
\[ \NS_{\delta}(f) = O_d(\, \epsilon^{1/(2d+2)}\,).\]

For $S \subseteq [n]$ and $x \in \dpm^n$ let $P_{x,S}:\dpm^{\bar{S}} \rgta \reals$ be the degree at most $d$ polynomial defined by $P_{x,S}(X_{\bar{S}}) = P(x_{|S},X_{\bar{S}})$.. Fix a $x = (x_1,\ldots,x_n) \in \dpm^n$ and define $S_{x,i} \subseteq [n]$ for $i \geq 1$, recursively as follows. $S_{x,1}$ is the set of $M_1 \leq L$ largest weight coordinates in $P$ given by applying \newref[Lemma]{mainlmns} to $P$. For $i \geq 1$, let $S^{x,i} = S_{x,1} \cup S_{x,2} \cup \ldots \cup S_{x,i}$.

For $i > 1$, let $S_{x,i+1}$ be the set of $M_{i+1} \leq L$ largest weight coordinates in $P_{x,S^{x,i}}$ given by applying \newref[Lemma]{mainlmns} to the polynomial $P_{x,S^{x,i}}$. Define $f_{x,i}$ by $f_{x,i}(\cdot) \equiv \sgn(P_{x,S^{x,i}}(\cdot)-\theta)$. Note that the definition of $f_{x,i}$ only depends on $x_j$ for $j \in S^{x,i}$ and that $|S^{x,i}| \leq L \cdot i$.

Call $x \in \dpm^n$ $(\epsilon,f)$-good if there exists an $i$, $1 \leq i \leq t$ such that $\NS_{\epsilon}(f_{x,i}) \leq  \Delta_d\, \epsilon^{1/(2d+2)}$  and let $t_x$ be such an $i$ for a $(\epsilon,f)$-good $x$. Then, from the definition of $f_{x,i}$ and \newref[Lemma]{mainlmns}, 
\begin{equation}
  \label{eq:3}
 \pr_{x \in_u \dpm^n}[\,\text{$x$ is $(\epsilon,f)$-good}\,] \geq 1-\epsilon.  
\end{equation}

Let $y \in_\delta x$ be a $\delta$-perturbation of $x \in_u \dpm^n$. Then, since $|S^{x,t_x}| \leq L\, t$, 
\begin{equation}
  \label{eq:4}
 \pr_{x,y}[\,x_{|S^{x,t_x}} \neq y_{|S^{x,t_x}}\,] \leq L\, t\, \delta = \epsilon^{1/(2d+2)}.  
\end{equation}

Also note that for any $i \geq 1$, conditioned on an assignment for the values in $x_{|S^{x,i}}$ and $x_{|S^{x,i}} = y_{|S^{x,i}}$, $\pr_{x,y}[f(x) \neq f(y)] = \NS_{\delta}(f_{x,i}) \leq \NS_\epsilon(f_{x,i})$. Thus, conditioned on $x$ being $(\epsilon,f)$-good and $x_{|S^{x,t_x}} = y_{|S^{x,t_x}}$, 
%\[ \pr_{x,y}\,[\,f(x) \neq f(y)\,|\,x_{|S^{x,t_x}} = y_{|S^{x,t_x}}\text{ and $x$ is $(\epsilon,f)$-good} \,] \leq  \NS_\epsilon(f_{x,t_x}) \leq c_2 \,\epsilon^{1/(2d+2)}.\]
\begin{equation}
  \label{eq:5}
 \pr_{x,y}\,[\,f(x) \neq f(y) \,] \leq  \NS_\epsilon(f_{x,t_x}) \leq \Delta_d \,\epsilon^{1/(2d+2)}.  
\end{equation}
Combining \eqref{eq:3}, \eqref{eq:4}, \eqref{eq:5}, we get
\[ \NS_{\delta}(f) \leq \epsilon + L\,t\,\delta+  \Delta_d\epsilon^{1/(2d+2)} = O_d\left(\,\epsilon^{1/(2d+2)}\,\right).\]
Since $\delta = O_d\left(\eps^{\frac{4d+5}{2d+2}}/\log^2(1/\epsilon)\right)$ and the above is applicable for all $\epsilon > 0$, we get that for all $\rho >0$,
\[ \NS_{\rho}(f) = O_d\left(\,\log(1/\rho)\rho^{1/(4d+5)}\,\right) = O_d\left(\,\rho^{1/(4d+6)}\,\right).\]
\end{proof}
\vspace{-.4in}
\section{Average sensitivity of PTFs}\label{sec:as}

In this section we bound the average sensitivity of PTFs on the Boolean hypercube, proving \newref[Theorem]{thm:as}. We first prove a lemma bounding the average sensitivity of a Boolean function in terms of its noise sensitivity. \tref{as} follows immediately from \tref{bns} and the following lemma:

\begin{lemma}[noise sensitivity to average sensitivity]\label{lm:nstoas}
  For any Boolean function $f:\dpmn \rgta \dpm$, $\AS(f) \leq 2 n e \,\NS_{(1/n)}(f)$.
\end{lemma}
\begin{proof}
Let $\delta = 1/n$. Let $X \in_u \dpmn$ and let $S \subseteq [n]$ be a random set with each element $i \in [n]$ present in $S$ independently with probability $\delta$. Let $X(S)$ be the vector obtained by flipping the coordinates of $X$ in $S$. Then, $\NS(f) = \Pr_{X,S}[f(X) \neq f(X(S))]$. Observe that for $i \in [n]$, $\Pr[ S = {i}] = \delta(1-\delta)^{n-1} = (1/n)\,(1-1/n)^{n-1} > 1/2ne$.
%\[\Pr[ S = {i}] = \delta(1-\delta)^{n-1} = \frac{1}{n}\,\left(1-\frac{1}{n}\right)^{n-1} > \frac{1}{2ne}.\]
Therefore,
\begin{align*}
\NS_\delta(f) &= \Pr_{X,S}[f(X) \neq f(X(S))] \\
&= \sum_i \Pr_S[\,S = \{i\}\,]  \cdot \Pr_X [\,f(X) \neq f(X(S))\,|\, S = {i}\,] + \Pr_S[\,|S| \neq 1\,] \cdot \Pr_{X,S}[\, f(X) \neq f(X(S))\,|\,|S| \neq 1\,]\\
&> \sum_i \frac{1}{2ne} \, \Pr_X[ f(X) \neq f(X(\{i\})] = \frac{1}{2ne} \AS(f).  
\end{align*}
\end{proof}

We now give a bound of $O(n^{1-2^{-d}})$ on the average sensitivity using a different (not using the noise sensitivity bounds), combinatorial, argument.
\begin{theorem}\label{combas} For any degree $d$ PTF $f: \{1,-1\}^n \to \{1,-1\}$, $\AS(f) \leq 3\, n^{1-2^{-d}}$.
\end{theorem}

We first show the theorem using \lref{mainlmc}. 
\begin{proof}
Let $P(x) = x_i P_i(x_{-i}) + Q_i(x_{-i})$, where $P_i(\;), Q_i(\;)$ are degree $d-1$ and degree $d$ polynomials respectively that do not depend on $x_i$. Define $f_i(x_{-i}) = \sgn(P_i(x_{-i}))$ and $g_i(x) = f(x) f_i(x_{-i})$. Then,
\begin{align*}\label{}
 \infl_i(f) &= \pr_{X \in_u \dpm^n}[f(X) \neq f(X^{(i)})]  = \pr_{X \in_u \dpm^n}[f(X) f_i(X_{-i}) \neq f(X^{(i)}) f_i(X_{-i})]\\
 &= \pr_{X \in_u \dpm^n}[f(X) f_i(X_{-i}) \neq f(X^{(i)}) f_i((X^{(i)})_{-i})] = \pr_{X \in_u \dpm^n}[g_i(X) \neq g_i(X^{(i)})]\\
 &= \infl_i(g_i).
\end{align*}

Observe that $g_i$ is monotone increasing in $x_i$ for $i \in [n]$ and hence $\infl_i(g_i) = \ex_X[X_ig_i(X)]$. Thus, 
\begin{multline*}
 \AS(f) = \sum_i \infl_i(f) = \sum_i \infl_i(g_i)
  = \sum_i \ex_X[X_i g_i(X)]
  = \sum_i \ex_X[X_i f(X) f_i(X_{-i})]
  = \ex_X \left[ f(X) \sum_i X_i f_i(X_{-i})\right].
\end{multline*}
Since $|f(x)| \leq 1$ for all $x$, we have
\begin{equation}\label{eq1:c}
  \AS(f) \leq \ex_X\left[\, \left| \sum_i X_i f_i(X_{-i})\right|\,\right].
\end{equation}

We now use induction and \lref{mainlmc}. For an LTF $f$, $f_i$ as defined above are constants. Therefore, by Equation \eqref{eq1:c},
\[ \AS(f) \leq  \ex_X\left[\, \left| \sum_i X_i f_i(X_{-i})\right|\,\right] = \ex_X\left[\left|\sum_i X_i\right|\right] = O(\sqrt{n}).\]
Suppose the theorem is true for degree $d$ PTFs and let $f$ be a degree $d+1$ PTF and let $f_i$ be as defined before. Then, by Equation \eqref{eq1:c} and \newref[Lemma]{mainlmc} 
\[ \AS(f)^2 \leq 2 \sum_i \AS(f_i) + n \leq \sum_i 6\,n^{1-2^{-d}} +n \leq 7\, n^{2 - 2^{-d}} .\]
Therefore, $\AS(f) \leq 3 \,n^{1-2^{-(d+1)}}$. The theorem follows by induction.
\end{proof}

\begin{proof}[Proof of \lref{mainlmc}]
For brevity, let $f_i(x) = f_i(x_{-i})$. By Cauchy-Schwarz, for any random variable $Z$ we have $\ex[|Z|]^2 \leq \ex[Z^2]$.  Thus,
\begin{align}\label{eq4}
 \ex_X\left[\, \left| \sum_i X_i f_i(X_{-i})\right|\,\right]^2 &\leq \ex_X\left[\, \left( \sum_i X_i f_i(X_{-i})\right)^2\,\right]\nonumber\\
  &= \ex_X [\,\sum_{i,j} X_i X_j f_i(X) f_j(X) \,]\nonumber\\
%  &= \ex_X[\, \sum_i 1\,] + \ex_X[\, \sum_{i\neq j} X_i X_j f_i(X) f_j(X)\,]\\
  &= n + \sum_{i \neq j}\,\ex_X[\,X_i X_j f_i(X) f_j(X)\,].
\end{align}
For $i \neq j \in [n]$, let $x_{-ij} = (x_k: k \in [n], k \neq i,j)$ and let $S_i^j = \{x \in \dpm^n : f_i(x) \neq f_i(x\oplus e_j)\}$. Note that $\infl_j(f_i) = \pr_X[X \in S_i^j]$. Now,
\begin{equation}\label{eq2:c}
  \ex_X[\,X_i X_j f_i(X) f_j(X)\,] = \sum_{x \in S_i^j \cup S_j^i}\, \mu(x)\, x_i x_j f_i(x) f_j(x) \; + \; \sum_{x \notin S_i^j \cup S_j^i}\mu(x)\, x_i x_j f_i(x) f_j(x),  
\end{equation}
where $\mu(x) = 1/2^n$ is the probability of choosing $x$ under the uniform distribution. We bound the first term in the above expression by the average sensitivity of the $f_i$'s and show that the second term vanishes. Observe that,
\begin{equation}\label{eq5}
  \sum_{x \in S_i^j \cup S_j^i} \mu(x) x_i x_j f_i(x) f_j(x) \leq \mu(S_i^j \cup S_j^i) \leq \mu(S_i^j) + \mu(S_j^i) = \infl_j(f_i) + \infl_i(f_j).  
\end{equation}

Note that for $x \notin S_i^j \cup S_j^i$, $f_i(x), f_j(x)$ are both independent of the values of $x_i,x_j$. For such $x$ (abusing notation) let $f_i(x_{-ij}) = f_i(x)$, $f_j(x_{-ij}) = f_j(x)$ and let $T_{ij} = \{(x_k:k \neq i,j): x \notin  S_i^j \cup S_j^i\}$. Then, since for $x \notin S_i^j \cup S_j^i$, $f_i(x),f_j(x)$ depend only on $x_{-ij}$, we get that $x \notin S_i^j \cup S_j^i$ if and only if $x_{-ij} \notin T_{ij}$. Therefore,
\begin{align}\label{eq6}
  \sum_{x \notin S_i^j \cup S_j^i}\mu(x)\, x_i x_j f_i(x) f_j(x) &= \sum_{x \notin S_i^j \cup S_j^i} \mu(x_{-ij})\,\mu(x_i)\, \mu(x_j) \, f_i(x_{-ij})\, f_j(x_{-ij}) \, x_i x_j\nonumber\\
&= \sum_{x_{-ij} \notin T_{ij}}\, \mu(x_{-ij}) f_i(x_{-ij})\, f_j(x_{-ij}) \,\ex_{x_i,x_j}[x_i x_j] = 0.
\end{align}

From Equations \eqref{eq4}, \eqref{eq2:c}, \eqref{eq5},\eqref{eq6} we have,
\begin{align*}\label{}
 \ex_X\left[\, \left| \sum_i X_i f_i(X_{-i})\right|\,\right]^2 \leq n + \sum_{i \neq j} (\infl_j(f_i) + \infl_i(f_j)) = n+2 \sum_i \sum_{j: j \neq i} \infl_j(f_i) =n+ 2 \sum_i \AS(f_i)  .
\end{align*}
\end{proof}

\begin{remark}
The bound of \newref[Lemma]{mainlmc} is tight up to a constant factor if we only have bounds on the average sensitivity of the $f_i$'s to go with. For example, consider $f_i$ defined as follows. Divide $[n]$ into $m = \sqrt{n}$ blocks $B_1,\ldots,B_m$ of size $m$ each and for $1 \leq j \leq m$, $i \in B_j$, let $f_i = \prod_{k \in B_j: k \neq i} x_k$. Then, the left hand side of the lemma is $\Theta(n^{3/2})$ and $\AS(f_i) = m - 1 = \Theta(\sqrt{n})$ for all $i$.
\end{remark}

\section*{Acknowledgments}
We thank Ilias Diakonikolas for pointing out an error in an earlier version of this writeup.

{\footnotesize
\bibliographystyle{prahladhurl}
\bibliography{ptf-bib}
%\bibliography{../../../jrnl-names-abb,../../../prahladhbib,../../../crossref}

}

\appendix
\section*{Appendix}
\section{Bounding the perturbation polynomial in the Gaussian setting} \label{sec:claimproof}
\subsection{Background on Hermite polynomials}

The univariate Hermite polynomials are defined as follows
$$H_k(x) = \frac{(-1)^k}{\sqrt{k!}}e^{x^2/2}\frac{d^k}{dx^k}e^{-x^2/2}.$$
The univariate Hermite polynomials satisfy $H'_k(x) = \sqrt{k} H_{k-1}(x)$.  

The multivariate Hermite polynomials over $n$ variables $(x_1,\dots,x_n)$ are defined as follows. Let $S\subseteq [d]$ be a multiset. It will be convenient to denote a multiset $S$ by a sequence of $n$ indices as $S=(s_1,\dots,s_n)$ where each $s_i$ denotes the cardinality of element $i\in [n]$ in the set $S$. Note, by this notation, $|S| = \sum s_i$.

$$H_S(x_1,\dots,x_n) = \prod_{i=1}^n H_{s_i}(x_i).$$

The partial derivatives of the multivariate Hermite polynomials can now be calculated as follows
$$(\partial H_S)_{i} (x_1,\dots,x_n) = \sqrt{s_i}H_{s_i-1}(x_i)\prod_{j\neq i} H_{s_j}(x_j) = \sqrt{s_i}H_{S\setminus\{i\}}(x_1,\dots,x_n).$$

Furthermore, the iterative partial derivatives can be calculated as follows. Let $R =(r_1,\dots,r_n)\subseteq S$ be any multiset.
$$(\partial H_S)_{R} = \sqrt{\prod_{i=1}^n \frac{s_i!}{(s_i-r_i)!}}\cdot H_{S\setminus R}.$$
This in particular gives the follow Taylor series expansion for $H_S(z)=H_S(z_1,\dots,z_n)$ about the point $x=(x_1,\dots,x_n)$ for multisets $S$. Let $|S|=d$. Since $H_S$ depends on at most $d$ variables, we can assume without loss of generality that $H_S$ is defined on the first $d$ variables, i.e., $S \subseteq [d]$ and $H_S(x)=H_S(x_1,\dots,x_d)$.
\begin{eqnarray}
H_S(z) &=& H_S(x) + \sum_{k=1}^d \sum_{R: |R| = k} \frac{1}{\prod_{i=1}^d r_i!}(\partial H_S)_R(x)\cdot \left( \prod_{i=1}^d (z_i-x_i)^{r_i}\right)\nonumber\\
&=&H_S(x) + \sum_{k=1}^d \sum_{R: |R| = k} \frac{1}{\prod_{i=1}^d r_i!}\sqrt{\prod_{i=1}^d \frac{s_i!}{(s_i-r_i)!}}\cdot H_{S\setminus R}(x)\cdot \left( \prod_{i=1}^d (z_i-x_i)^{r_i}\right)\label{eq:taylor}
\end{eqnarray}

The multivariate Hermite polynomials up to degree $d$ form a basis for the set of all multivariate polynomials of degree $d$. In particular, given any degree $d$ polynomial $P(x_1,\dots,x_n)$, we can write it as a linear combination of Hermite polynomials as follows
\begin{eqnarray*}
P(x_1,\dots,x_n) & = & \sum_{S\subset [n]: |S| \leq d} \hat{P}_S \,H_S(x_1,\dots,x_n)
\end{eqnarray*}
The values $\hat{P}_S$ are called the Hermite coefficients of $P$.

The Hermite polynomials are especially useful while working over the (multivariate) normal distribution due to the following orthonormality conditions.
\begin{eqnarray*}
\ex_{X \gets \NN} \left[ H_S(X) H_T(X) \right] &= &
\begin{cases}
1 & \text{if } S = T\\
  0 & \text{otherwise.}
\end{cases}
\end{eqnarray*}                   
This implies that $\|P\|^2 = \ex_{X\gets \NN}[P^2(X)] = \sum \hat{P}^2_S$.

\subsection{Proof of {\newref[Claim]{claim:smallQ}}}

Recall that we must prove there exists a constant $c_d$ such that $\|Q\| \leq c_d\sqrt{\delta}$.

\begin{proof}
Given any degree $d$ multivariate polynomial $P$, we can write it in the Hermite basis as $P(x) = \sum_{S: |S| \leq d} \hat{P}_S H_S(x)$ and use this expansion to bound $\|Q\| = \|P(Z)-P(X)\|$ as follows.
\begin{align}
\|Q\|^2&=\ex\left[(P(Z)-P(X))^2\right] = \ex\left[ \left(\sum_{S: |S| \leq d} \hat{P}_S (H_S(Z)-H_S(X))\right)^2\right]\nonumber\\
& =  \sum_{S,T}\hat{P}_S\hat{P}_T \ex\left[(H_S(Z)-H_S(X))\cdot (H_T(Z)-H_T(X)) \right]\nonumber\\
&=\sum_{S}\hat{P}^2_S \ex\left[(H_S(Z)-H_S(X))^2\right] + \sum_{S\neq T} \hat{P}_S\hat{P}_T \ex\left[(H_S(Z)-H_S(X))\cdot (H_T(Z)-H_T(X)) \right]\nonumber\\
& = \sum_{S}\hat{P}^2_S \ex\left[(H_S(Z)-H_S(X))^2\right] -\sum_{S\neq T} \hat{P}_S \hat{P}_T \left(\ex\left[H_S(Z)H_T(X)\right] + \ex\left[H_S(X)H_T(Z)\right]\right)\label{eq:qbound}
\end{align}
where the last step follows from the orthonormality of the Hermite polynomials. 

We will now show that $\ex[H_S(X)H_T(Z)]=0$ for $S \neq T$. Since $S\neq T$, it suffices to show the following univariate case: $\ex[H_s(X)H_t(Z)] =0$ for $s\neq t$. We now observe that the joint distribution $(X,Z)$ is identical to the distribution $(Z,X)$. Hence, to calculate $\ex[H_s(X)H_t(Z)]$ for $s \neq t$ we can assume without loss of generality that $s > t$. Now, $H_t(Z)=H_t((1-\delta)X+\rho Y)$ is a bivariate degree $t$ polynomial and can be expanded in the Hermite basis as $\sum_{i,j =0}^t \alpha_{ij}H_i(X)H_j(Y)$. We thus have $\ex_{X,Y}[H_s(X)H_t(Z)] = \sum_{i,j=0}^t \alpha_{ij}\ex_X[H_s(X)H_i(X)]\cdot \ex_Y[H_j(Y)]=0$ since $s > t \geq i$.

Plugging this into the expression for $\|Q\|$ in \eqref{eq:qbound}, we have $\|Q\|^2=\sum_{S}\hat{P}^2_S \ex\left[(H_S(Z)-H_S(X))^2\right] $. Since $\|P\|^2= \sum_S \hat{P}^2_S=1$, to prove the claim it suffices if  we show that there exists a constant $c_d$ such that for any multiset $S$, $\|H_S(Z)-H_S(X)\|^2 \leq c^2_d\delta$. We bound the norm $\|H_S(Z)-H_S(X)\|$ using the Taylor series expansion of $H_S(Z)$ as stated in equation \eqref{eq:taylor}. Let $|S|=d$; then we have

\begin{align*}
\ex_{X,Y}\left[ |H_S(Z)-H_S(X)|\right] 
& \leq \sum_{k=1}^d \sum_{R:|R|=k} \frac{1}{\prod_{i=1}^d r_i!}\sqrt{\prod_{i=1}^d \frac{s_i!}{(s_i-r_i)!}}\cdot\ex\left[\left| H_{S\setminus R}(X)\cdot \left( \prod_{i=1}^d (Z_i-X_i)^{r_i}\right)\right|\right]\\
& \leq \sum_{k=1}^d \sum_{R:|R|=k} \frac{1}{\prod_{i=1}^d r_i!} d^{k/2} \ex \left[\left|H_{S\setminus R}(X)\cdot \left( \prod_{i=1}^d (Z_i-X_i)^{r_i}\right)\right|\right]\\
& \qquad \qquad \text{[ since each $s_i \leq d$ and $\sum r_i =|R|=k$ ]}\\
& \leq \sum_{k=1}^d d^{k/2}\sum_{R:|R|=k} \frac{1}{\prod_{i=1}^d r_i!} \sqrt{\ex\left[H^2_{S\setminus R}(X)\right]\cdot \ex\left[\prod_{i=1}^d (Z_i-X_i)^{2r_i}\right]}\\
& \qquad \qquad \text{[ By Cauchy-Schwarz inequality ]}\\
&=\sum_{k=1}^d d^{k/2}\sum_{R:|R|=k} \frac{1}{\prod_{i=1}^d r_i!} \sqrt{\prod_{i=1}^d\ex\left[(Z_i-X_i)^{2r_i}\right]}\\
& \qquad \qquad \text{[ By orthonormality of $H_{S\setminus R}$ and independence of $(Z_i-X_i)$ over the $i$'s ]}\\
&= \sum_{k=1}^d d^{k/2}\sum_{R:|R|=k} \frac{1}{\prod_{i=1}^d r_i!} \sqrt{\prod_{i=1}^d\delta^{r_i} \frac{(2r_i)!}{r_i!}}\\
& \qquad \qquad \text{ [ Since $Z_i-X_i \sim \calN(0,\sqrt{2\delta})$ whose $2r$-th moment is $\delta^r \frac{(2r)!}{r!}$ ]}\\
& = \sum_{k=1}^d (d\delta)^{k/2}\sum_{R:|R|=k} \sqrt{\prod_{i=1}^d \binom{2r_i}{r_i}} \leq \sum_{k=1}^d (d\delta)^{k/2} \cdot d^k \cdot 2^{k/2} = \sum_{k=1}^d (d^{3/2}\sqrt{2\delta})^k\\
& \leq 2d^{3/2}\sqrt{2\delta} \qquad \qquad \text{[ If $d^{3/2}\sqrt{2\delta} \leq 1/2$ ]}
\end{align*}
Thus, if $d^{3/2}\sqrt{2\delta} \leq 1/2$, then $\ex[|H_S(Z)-H_S(X)|] \leq 2d^{3/2}\sqrt{2\delta}$.
We can now use $(1,2)$-hypercontractivity for degree $d$ polynomials under the normal distribution (see \cite[Remark~5.13]{Janson1997}), and bound $\|H_S(Z)-H_S(X)\|$ as follows.$$\|H_S(Z)-H_S(X)\|^2\leq e^d \ex\left[ |H_S(Z)-H_S(X)|\right] \leq 2d^{3/2}e^d\sqrt{2\delta}.$$

If $d^{3/2}\sqrt{2\delta} > 1/2$, we have $\ex[|H_S(Z)-H_S(X)|^2] \leq 2\ex[H_S^2(Z) + H_S^2(X)] \leq 4 < 8d^{3/2}\sqrt{2\delta}$. Thus, either way, we have that there exists a constant $c_d$ such that $\|H_S(Z)-H_S(X)\| \leq c_d\sqrt{\delta}$.  
\end{proof}

%%% Local Variables: 
%%% mode: latex
%%% TeX-master: "ptf"
%%% End: 

\end{document}